\documentclass[10pt,a4paper]{article}

\usepackage{amsmath}
\usepackage{amsthm}
\usepackage{amssymb}
\usepackage{bm}
\usepackage{ascmac}
\usepackage[dvipdfmx]{graphicx}
\newtheorem{theorem}{Theorem}

\begin{document}

\title{Exact Lyapunov exponents of the generalized Boole transformations}


\author{Ken Umeno and Ken-ichi Okubo}

\maketitle

\begin{abstract}%
The generalized Boole transformations have rich behavior ranging from  the 
\textit{{mixing}} phase with the Cauchy {invariant} measure to the \textit{dissipative} phase through
the \textit{infinite ergodic} phase with the Lebesgue measure. 
In this Letter, {by giving the proof of mixing property for $0<\alpha<1$}
we show an {\it analytic} formula of the Lyapunov exponents $\lambda$ 
{which are explicitly parameterized in terms of the parameter $\alpha$ of the generalized Boole transformations 
	for {the} whole region $\alpha>0$
	and bridge those three phase {\it continuously}. 
	We found the different scale behavior of the Lyapunov exponent near $\alpha=1$ using analytic formula with the parameter $\alpha$.
	In particular, for $0<\alpha<1$, we  then prove an existence of extremely sensitive dependency of Lyapunov exponents, where
	the absolute values of the derivative of Lyapunov exponents 
	with respect to the parameter $\alpha$ diverge to infinity  
	in the limit of $\alpha\to 0$,  and $\alpha \to 1$.} This result shows the computational complexity on the numerical simulations of 
the Lyapunov exponents near $\alpha \simeq$ 0, 1.  
\end{abstract}

\paragraph{1. Introduction-Generalized Boole transformations}
Lyapunov exponent has  extensively studied for characterizing chaotic systems and plays an important and essential role as {a} chaos indicator
characterizing the degree of chaoticity quantitatively.  However, one can find it difficult to obtain an analytic formula of Lyapunov exponents 
for given chaotic systems in general. There are few analytic formulas of ergodic {invariant measures} \cite{Umeno97}. 
The analytic {parametric density functions of ergodic invariant measure} can be used for analytically estimating Lyapunov exponents 
for such specific systems. 
However, {for most cases,} that analytic result cannot be extended for the whole region of parameter space. 
Consequently, we are forced to use of numerical simulations to research {a} parameter dependence of Lyapunov exponents.
In the case of the logistic map, Huberman and Rudnick \cite{Huberman} found a scaling behavior of Lyapunov exponent from zero to positive and 
Pomeau and Manneville \cite{Pomeau} found a similar behavior of Lyapunov exponent. Both of them use numerical simulations.

Here, we consider the Boole transformation \cite{Boole} given by 
\begin{eqnarray}
	x_{n+1} = Tx_n = x_n - \frac{1}{x_n},
\end{eqnarray}
which is known to preserves the Lebesgue measure $dx$ as an invariant measure \cite{Adler}. 
The Lebesgue measure $dx$ is an infinite measure. 
{Thus, we cannot normalize the invariant measure as a probability measure.}

One can generalize the Boole transformation as {two-parameter map family as} Aaronson \cite{Aaronson97} 
treats the {following} generalized Boole transformations which are defined as
\begin{eqnarray}
	x_{n+1} = {T_{\alpha, \beta}}x_n=\alpha x_n -\frac{\beta}{x_n},~{\alpha>0}, \beta>0,  \label{Generalized Boole}
\end{eqnarray}
preserving the Cauchy measure $\mu({dx})$ {for $0<\alpha<1$ given by}
\begin{eqnarray}
	\mu({dx}) = \frac{\sqrt{\beta(1-\alpha)}}{\pi[x^2(1-\alpha)+\beta]}{dx}. \label{Cauchy measure}
\end{eqnarray}
As a subset of the generalized Boole transformations, the map at $\alpha=\beta=1/2$ corresponds to the exactly solvable chaos map \cite{Umeno97}
given by  Ref. \cite{Umeno98}
\begin{eqnarray}
	x_{n+1}=\frac{1}{2}\left(x_n-\frac{1}{x_n} \right). \label{Umeno map}
\end{eqnarray}
One of the authors showed that this map can be expressed by the double-angle formula of cot function having Lyapunov 
exponent $\lambda =\log 2$ and one can use
Eq.(\ref{Umeno map}) in order to generate general L\'evy stable distributions via their superposition \cite{Umeno98}.

{The generalized Boole transformations have the standard forms which do not depend on $\beta$. That is, by substituting 
	$x_n$ by $\sqrt{\frac{\beta}{\alpha}}y_n$, Eq. (\ref{Generalized Boole}) is {replaced by}
	\begin{eqnarray}
		\sqrt{\frac{\beta}{\alpha}}y_{n+1} &=& \alpha \sqrt{\frac{\beta}{\alpha}}y_n-\beta \sqrt{\frac{\alpha}{\beta}}\frac{1}{y_n}, \nonumber
	\end{eqnarray}
	{which is nothing but}
	\begin{eqnarray}
		y_{n+1} &=& \alpha\left(y_n-\frac{1}{y_n}\right). \label{modified Generalized Boole}
	\end{eqnarray}
	The transformation $P_{\alpha, \beta} : x_n \mapsto \sqrt{\frac{\beta}{\alpha}}y_n$ is {a one-to-one diffeomorphism}. Then the transformations of 
	Eq. (\ref{modified Generalized Boole}) have the same Lyapunov exponent for fixed $\alpha$ with this topological conjugacy relation. 
	Consequently, the generalized Boole transformations
	have the same Lyapunov exponents which are independent of $\beta$.}

\paragraph{2. {Analytic} Formula of Lyapunov exponents for the generalized Boole transformations}
Here, we prove the following theorem. 
\begin{theorem}\label{T0.1}
	The generalized Boole transformations {for $0<\alpha<1$} have the Lyapunov exponents given by 
	\begin{eqnarray*}
		\lambda(\alpha, \beta) = \log \left(1 + 2\sqrt{\alpha(1-\alpha)}\right). 
	\end{eqnarray*}
\end{theorem}

Note that the Lyapunov exponents do not depend on the parameter $\beta$ but purely depend on the parameter $\alpha$.

\begin{proof}
	The ergodicity of {similar kinds of this} generalized Boole transformations was shown by  Kempermann \cite{Kempermann} {and Letac \cite{Letac}}.
	{We prove the theorem 2 concerning the mixing property of the generalized Boole transformations for $0<\alpha<1$ in the Appendix. 
		\begin{theorem}
			The generalized Boole transformations in $0<\alpha<1$ {have mixing property}.
		\end{theorem}
		Thus, we can prove the ergodicity with respect to the Cauchy measure in Eq. (\ref{Cauchy measure}) for $0<\alpha<1$.}
	By the ergodicity, we can calculate the Lyapunov exponents of the generalized Boole transformations as
	\begin{eqnarray*}
		\lambda(\alpha,\beta) 
		= \int_{-\infty}^{\infty} \log\left|\alpha+\frac{\beta}{x^2}\right| \frac{\sqrt{\beta(1-\alpha)}}{\pi[x^2(1-\alpha)+\beta]} dx.
	\end{eqnarray*}
	
	Here,  we substitute~ $\sqrt{\alpha/\beta}x$~ by~ $y$, then we obtain:
	\begin{eqnarray*}
		\lambda(\alpha,\beta) 	&=& \frac{1}{\pi}\sqrt{\frac{\alpha}{1-\alpha}}\int_{-\infty}^{\infty} \log\left|\alpha+\frac{\alpha}{y^2}\right|
		\frac{dy}{y^2+\left( \sqrt{\frac{\alpha}{1-\alpha}}\right)^2 },\\
		\log\left|\alpha+\frac{\alpha}{y^2}\right| &=& \log \alpha + \log|1+iy|+\log|1-iy|-2\log|y|.
	\end{eqnarray*}
	Let us substitute $p$ for $\sqrt{\frac{\alpha}{1-\alpha}}$ and put $I_1,\ I_2,\ I_3$ and $I_4$ as,
	\begin{eqnarray*}
		\lambda(\alpha,\beta) &=& I_1 + I_2 +I_3+I_4,\\
		I_1 &=& \frac{p}{\pi}\int_{-\infty}^{\infty}  \frac{\log \alpha}{y^2+p^2}dy, \\
		I_2 &=& \frac{p}{\pi}\int_{-\infty}^{\infty}  \frac{\log|1+iy|}{y^2+p^2}dy, \\
		I_3 &=& \frac{p}{\pi}\int_{-\infty}^{\infty}  \frac{\log|1-iy|}{y^2+p^2}dy, \\
		I_4 &=& -\frac{2p}{\pi}\int_{-\infty}^{\infty}  \frac{\log|y|}{y^2+p^2}dy. 
	\end{eqnarray*}
	
	Evaluate $I_1$. We consider a function $f_1(z)$ defined by
	\begin{eqnarray*}
		f_1(z) = \frac{\log \alpha}{z+ip}, 
	\end{eqnarray*}
	which is regular on upper half-plane. Thus we can integrate along the path C1 in Figure \ref{integral1}. Then we have the relations
	\begin{figure}
		\begin{center}
			\includegraphics[width=6.5cm]{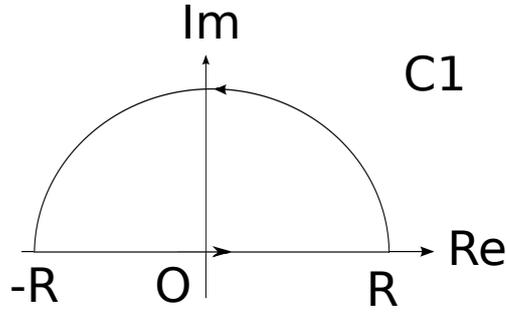}
			\caption{The integral path of $I_1, I_2 ~\mbox{and}~ I_3$}
			\label{integral1}
		\end{center}
	\end{figure}
	\begin{eqnarray}
		I_1 = \log \alpha,
	\end{eqnarray}
	as $R\to\infty$.
	
	Next, as for {$I_3$},  we consider a function {$f_{3}(z)$} defined by
	\begin{eqnarray*}
		{f_3}(z) = \frac{\log|1-iz|}{z+ip},
	\end{eqnarray*}
	which is regular on upper half-plane. Thus we can integrate {$f_{3}(z)$} along the path C1 in Figure \ref{integral1}.  We get: 
	\begin{eqnarray}
		{I_{3}} = \log|1+p|,
	\end{eqnarray}
	as $R\to\infty$.
	
	In terms of evaluating {$I_{2}$}, by applying a transformation as $y\to-y$ we have the relation {$I_{2}=I_{3}$}.
	\begin{figure}
		\begin{center}
			\hspace*{1cm}
			\includegraphics[width=6.5cm]{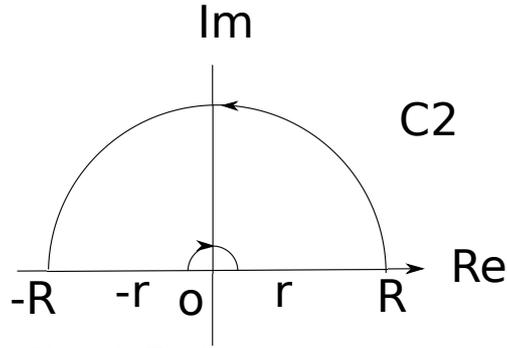}
			\vspace*{-5mm}
			\caption{The integral path of $I_4$}
			\label{integral2}
		\end{center}
	\end{figure}

	As for $I_4$, we consider a function $f_4(z)$ defined by
	\begin{eqnarray}
		f_4(z) = \frac{\log|z|}{z+ip}
	\end{eqnarray}
	which is regular along the path C2 in Figure \ref{integral2} and on the inner part of C2. Then by integrating of $f_4(z)$ along C2, we obtain
	\begin{eqnarray}
		I_4 = -2\log p,
	\end{eqnarray}
	as $r\to0, R\to \infty$. As for $r\to 0$, we substitute $z$ by $e^{i\theta}$ and we have the relation:
	\begin{eqnarray*}
		\left|\int_{\pi}^{0}\frac{\log r + 2\pi}{r^2e^{2i\theta}-p^2}ire^{i\theta}d\theta \right|
		\leq \int_{0}^{\pi}\frac{\left|\log r + 2\pi\right|}{\left|r^2e^{2i\theta}-p^2\right|}r d\theta.
	\end{eqnarray*}
	Then we set $r<p/2$ and we obtain:
	\begin{eqnarray*}
		\frac{1}{\left|r^2e^{2i\theta}+p^2\right|}<\frac{1}{|p|^2-r^2}<\frac{4}{3}\frac{1}{|p|^2}.
	\end{eqnarray*}
	Using these facts, we have
	\begin{eqnarray*}
		\int_{0}^{\pi}\frac{\left|\log r + 2\pi\right|}{\left|r^2e^{2i\theta}-p^2\right|}r d\theta
		&<& \int_{0}^{\pi} \frac{4}{3}\frac{\left|\log r + 2\pi\right|}{|p|^2}rd\theta,\\
		&\leq& \frac{4\pi}{3}\frac{r\left|\log r +2\pi\right|}{|p|^2} \to 0~{\rm as}~r\to0.
	\end{eqnarray*}
	
	Therefore, we obtain analytic formula of the Lyapunov exponents for the generalized Boole transformations $\lambda(\alpha,\beta)$ as
	\begin{eqnarray}
		\lambda(\alpha,\beta) &=& \log \alpha+\log|1+p|+\log|1+p|-2\log p,\nonumber\\
		&=& \log\left( 1+2\sqrt{\alpha(1-\alpha)}\right) \label{Lyapnov指数}.
	\end{eqnarray}
\end{proof}

The Lyapunov exponents do not depend on $\beta$, so that we can change the notation from $\lambda(\alpha, \beta)$ to $\lambda(\alpha)$. 
The invariant measure {$\mu(dx)$ is Cauchy} and smooth with respect to the Lebesgue measure.  Thus,  according to the Pesin identity \cite{Pesin,Ruelle}, we can calculate KS-entropy $h(\alpha)$ analytically as
\begin{eqnarray}
	h(\alpha) = \int_{-\infty}^{\infty}\lambda(\alpha)  {\mu(dx)}=\lambda(\alpha).
\end{eqnarray}

Figure \ref{LyapunovB} shows the comparison results of our numerical estimate of the Lyapunov exponents for the generalized Boole transformations with 

\noindent $\alpha =0.1,\ 0.2,\ 0.3,\ 0.4,\ 0.5,\ 0.6,\ 0.7,\ 0.8$ and $0.9$ using the long double precision and our analytic formula of the Lyapunov exponents. The initial condition of computation is $x_0=9.3$. The coincidence is remarkable as shown in Figure \ref{LyapunovB}.
\begin{figure}[h]
	\centering
	\vspace*{-1cm}
	\includegraphics[width=8cm]{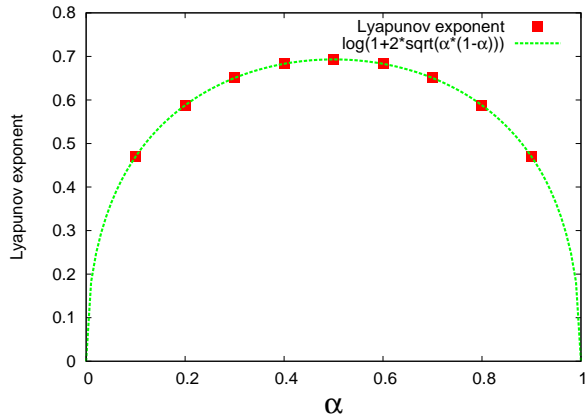}
	\vspace{1cm}
	\caption{The numerical results of the finite Lyapunov exponents and analytic solution of the generalized Boole transformation.}
	\label{LyapunovB}
\end{figure}

{
	Equation (\ref{Lyapnov指数}) converges to zero in the limit of $\alpha \to 0$ and $\alpha \to 1$. However, different phenomena occur {for} $\alpha=0$
	and $\alpha=1$ {respectively}. In the case of $\alpha=0$, every orbits are periodic because
	\begin{eqnarray}
		T^2_{\alpha, \beta}x=x. \nonumber
	\end{eqnarray}
	On the other hand, at $\alpha=1$, {\textit{subexponential chaos} \cite{Akimoto15,Akimoto}} occurs. 
	Then, in calculation of finite Lyapunov exponents, the value does not converge to a
	constant but converges in distribution. 
}

\paragraph{3. Lyapunov exponents {for} $\alpha>1$}
Here, we prove another theorem. 
\begin{theorem}\label{T0.2}
	The generalized Boole transformations {for $\alpha>1$} have the Lyapunov exponents given by 
	\begin{eqnarray*}
		\lambda(\alpha, \beta) = \log \alpha. 
	\end{eqnarray*}
\end{theorem}
\begin{proof}
	In the range of $\alpha>1$, the generalized Boole transformations have two equilibrium points at $x_\pm=\pm \sqrt{\frac{\beta}{\alpha-1}}$.
	Two of them are unstable since {the slope of the liner} inclinations are both $\frac{dT_{\alpha, \beta}x}{dx}(x_\pm)= 2\alpha-1>1$. Now by substituting $x_n$ by $\frac{1}{z_n}$,
	we obtain another form of the generalized Boole transformations as
	\begin{eqnarray}
		z_{n+1} = \frac{z_n}{\alpha-\beta z_n^2}. \label{medified GB2}
	\end{eqnarray}
	The equilibrium points of Eq. (\ref{medified GB2}) are $z=0, \pm\sqrt{\frac{\alpha-1}{\beta}}$. Two points $z= \pm\sqrt{\frac{\alpha-1}{\beta}}$ are unstable.
	The equilibrium point at $z=0$ are stable and attractive because 
	\begin{eqnarray}
		\left|\left. \frac{d}{dz} \left(\frac{z}{\alpha-\beta z^2}\right)\right|_{z=0}\right|<1. 
	\end{eqnarray}
	{Thus}, for almost all initial points it is hold that
	\begin{eqnarray}
		\lim_{n \to \infty} |x_n| = \infty.
	\end{eqnarray}
	Hence, the Lyapunov exponents of the generalized Boole transformations for $\alpha>1$ are given by taking logarithmic function 
	of inclination in the limit of $|x|\to \infty$ as
	\begin{eqnarray}
		\lambda(\alpha, \beta) = \log \alpha.
	\end{eqnarray}
\end{proof}
The complete behavior of Lyapunov exponent as a function of parameter $\alpha$ is shown in Figure \ref{Fig: Lyapunov Exponent}.
\begin{figure}[b]
	\centering
	\hspace*{-1cm}
	\includegraphics[width=.6\columnwidth]{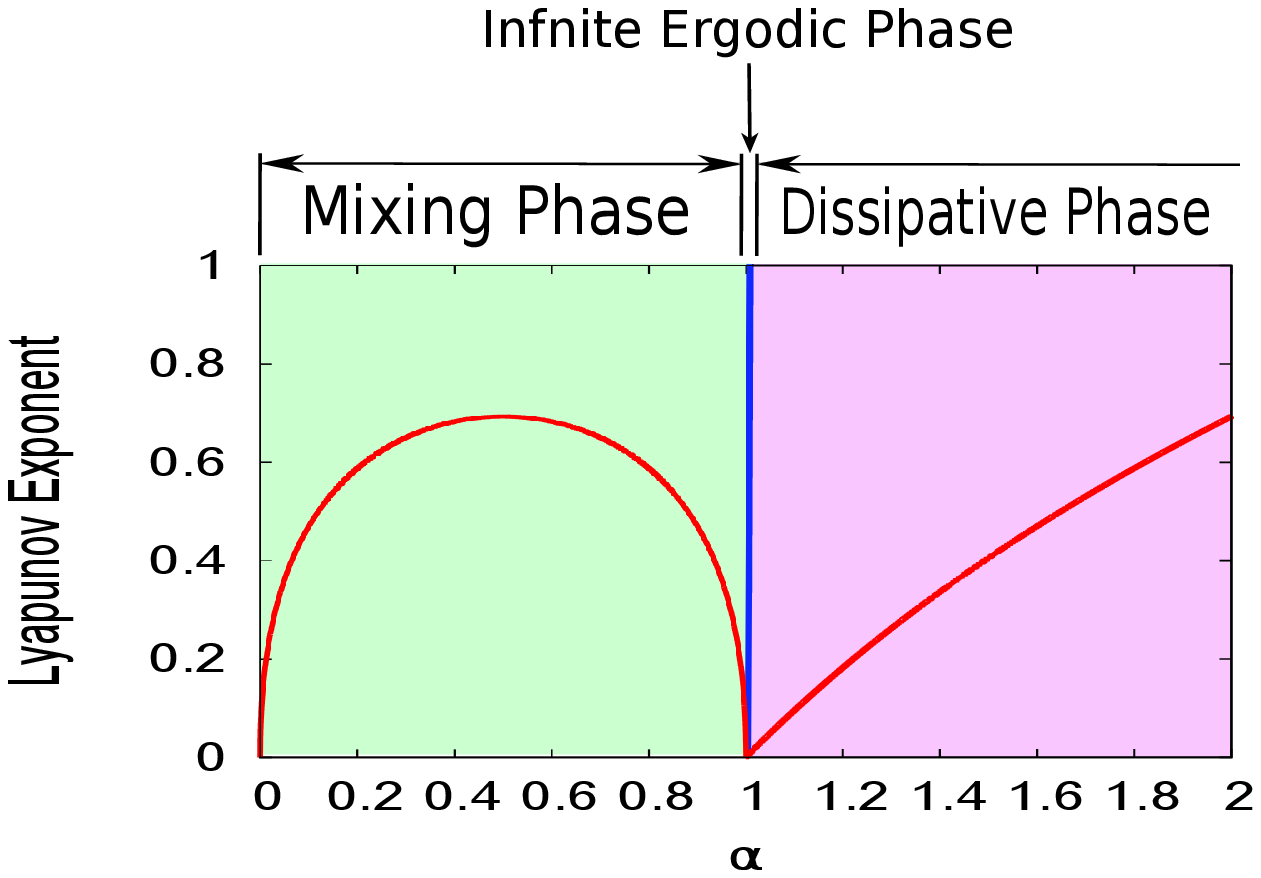}
	\caption{The classification of the phase of the {generalized Boole transformation systems in the parameter space $\alpha$
			via the analytic Lyapunov exponents.}}
	\label{Fig: Lyapunov Exponent}
\end{figure}
According to Ref. \cite{Aaronson97}, the generalized Boole transformations for $\alpha>1$ are \textit{dissipative}
in the sense that, for almost all initial points, $|x_n| \to \infty$ as $n\to \infty$.

Table \ref{Table: classification} summarizes a classification of the phase of systems by {the nature of} invariant measure. 
We can say the system is in \textit{Mixing Phase} for $0<\alpha<1$, 
\textit{Infinite Ergodic Phase} for $\alpha=1$ and \textit{Dissipative Phase} at $\alpha>1$.
\begin{table}[tb]
	\centering
	\caption{	{The classification of phase of dynamical systems, invariant measure and Lyapunov exponents
			in terms of parameter $\alpha$.}}
	\label{Table: classification}
	\renewcommand \arraystretch{1.5}
	
	\begin{tabular}{|c||c|c|c|} \hline
		& $0<\alpha<1$ & $\alpha=1$ & $\alpha>1$ \\ \hline
		Phase & {Mixing} Phase & Infinite Ergodic Phase &  Dissipative Phase\\ \hline
		Invariant measure	$\mu(dx)$	& $\frac{\sqrt{\beta(1-\alpha)}}{\pi[x^2(1-\alpha)+\beta]}dx$ & $dx$ & \\ \hline
		Lyapunov exponents & $\log\left( 1+2\sqrt{\alpha(1-\alpha)}\right)$ & 0 (\textit{subexponential chaos}) & $\log \alpha$ \\ \hline
	\end{tabular}
\end{table}
\begin{table}[tb]
	\centering
	\caption{{The difference between types of Intermittency and critical exponents of the generalized Boole transformation.}}
	\label{Table: classification-CE}
	\renewcommand \arraystretch{1.5}
	
	\begin{tabular}{|c||c|c|c|} \hline
		& $\alpha \to 0+0$ & $\alpha\to 1-0$ & $\alpha \to1+0$ \\ \hline
		Intermittency & \textit{Type 3} & \textit{Type 1} &  \\ \hline
		Critical Exponent & $\frac{1}{2}$ &  $\frac{1}{2}$ & 1  \\ \hline
	\end{tabular}
\end{table}

The standard Cauchy density function $\rho_\gamma(x)$ is defined by a scale parameter $\gamma$
\begin{eqnarray}
	\rho_\gamma(x) = \frac{\gamma}{\pi(x^2 + \gamma^2)}.
\end{eqnarray}
According to the probability preservation relation \cite{Umeno15}, the density functions $\rho_\gamma(x)$ and $\rho_{\gamma'}(y)$ satisfy the relation
\begin{eqnarray}
	\rho_{\gamma'}(y)dy = \sum_{x = T_{\alpha, \beta}^{-1}y} \rho_\gamma (x) dx.
\end{eqnarray}
In the case of $\alpha=\beta= \frac{1}{2}$ \cite{Umeno15}, the scale parameter $\gamma$ of 
Cauchy density function of invariant measure is updated in accordance with
\begin{eqnarray}
	\gamma' =\frac{1}{2}\left(\gamma+ \frac{1}{\gamma}\right).
\end{eqnarray}
Then, $\gamma$ converges to unity. 
Similarly, by the probability preservation relation, the scale parameter of the density function of the generalized Boole transformation is updated in accordance with
\begin{eqnarray}
	\gamma'  = \alpha \gamma + \frac{\beta}{\gamma},
\end{eqnarray}
and it converges to the fixed point $\gamma^*=\sqrt{\frac{\beta}{1-\alpha}}$. \\
As we make $\alpha$ close to unity from below, the scale parameter $\gamma$ converges to infinity and the density function of invariant measure 
becomes the Lebesgue measure $dx$.\\
For $\alpha>1$, the fixed point of scale parameter $\gamma^*$ must be an imaginary number. This fact shows for $\alpha>1$, the density function of 
the invariant measure can \textit{never} be a Cauchy distribution.

{
	Compare with the AL map \cite{Nakagawa14, Akimoto15} whose Lyapunov exponent diverges to infinity at the fixed point, 
	the generalized Boole transformations for $\alpha>1$ have the positive Lyapunov exponent $\log \alpha$ at the attractive fixed point $z=0$.}

\paragraph{4. Scaling behavior at $\alpha=0, 1$}
With parametric analytic formula of Lyapunov exponents, we can investigate a parametric dependency. 
The first derivative of the Lyapunov exponents with respect to $\alpha$ is given by
{
	\begin{eqnarray}
		\frac{d\lambda}{d\alpha} = 
		\begin{cases}
			\frac{1-2\alpha}{\sqrt{\alpha(1-\alpha)}\left(1+2\sqrt{\alpha(1-\alpha)}\right)}, & 0<\alpha<1, \\
			~~~~~~~~~\frac{1}{\alpha}, & \alpha>1.
		\end{cases}
	\end{eqnarray}
}
Remarkably, it diverges in the limit of $\alpha \to 0, {\alpha \to 1-0}$ as
{
	\begin{eqnarray}
		\left| ~\left.\lim_{t\to 0} \frac{d\lambda}{d\alpha}\right|_{\alpha=t} ~\right|= \infty,~~~
		\left|  ~\left.\lim_{t\to 1-0} \frac{d\lambda}{d\alpha}\right|_{\alpha=t} ~\right|= \infty.
	\end{eqnarray}
	This is similar to the well-known critical phenomena that the derivative of magnetization diverges at the critical temperature.
}
This means that a \textit{slight} modification of parameter $\alpha$ 
{toward unity (the Boole transformations) from below} causes a \textit{large} effect on 
the value of Lyapunov exponent {at the edge of $0<\alpha<1$}. Furthermore, we can say that  
it is {\it extremely} difficult to numerically obtain the Lyapunov exponents near {$\alpha=1$ and $0<\alpha<1$ because 
	for $\alpha=1$ \textit{subexponential chaos} \cite{Akimoto,Akimoto15}
	occurs and we cannot determine a Lyapunov exponent for finite calculation.}


{According to Eq. (\ref{modified Generalized Boole})}, 
the generalized Boole transformation at the limit of $\alpha\to 1$ corresponds to the Boole transformation, such that
a transition behavior between chaotic behavior and \textit{subexponential} behavior \cite{Akimoto} can be observed at $\alpha =1$.

Consider the scaling behavior at $\alpha=0, 1$. When {$\alpha\simeq0$ or ``$\alpha\simeq1$ and $0<\alpha<1$''	
	, the relation $0<2\sqrt{\alpha(1-\alpha)} \ll 1$ holds}.
{Thus, we obtain} a Taylor expansion {of Eq. (\ref{Lyapnov指数})} as
\begin{eqnarray}
	\lambda(\alpha) = \sum_{n=1}^{\infty} \frac{(-1)^{n+1}}{n}\left(2\sqrt{\alpha(1-\alpha)}\right)^n. 
\end{eqnarray}
Thus, the first approximation is given by the first term of the Taylor expansion
{
	\begin{eqnarray}
		\lambda &\simeq& 2\sqrt{\alpha(1-\alpha)},\\
		&\simeq& 
		\left\lbrace 
		\begin{array}{ccc}
			2\sqrt{\alpha}, & {\rm for}&\alpha \simeq 0,\\
			2\sqrt{1-\alpha}, & {\rm for}& \alpha \simeq 1~ {\rm and}~ 0<\alpha<1 .
		\end{array}\right.
	\end{eqnarray}
	Then, the scaling behavior of Lyapunov exponents near $\alpha=0,1$ is shown as
	\begin{eqnarray}
		\lambda &\simeq& \left\lbrace 
		\begin{array}{ccc}
			2 \alpha^{\frac{1}{2}}, & {\rm for}& \alpha \simeq 0,\\
			2(1-\alpha)^{\frac{1}{2}}, & {\rm for}& \alpha \simeq 1~{\rm and}~0<\alpha<1,\\
			\alpha-1, & {\rm for}& \alpha \simeq 1~{\rm and}~\alpha>1.
		\end{array}\right.
	\end{eqnarray}
}
{
	In this paper, we define a critical exponent as a power exponent of Lyapunov exponent. 
	A critical exponent $\delta$ is defined by
	\begin{eqnarray}
		\lambda \simeq C|r-r_c|^\delta,
	\end{eqnarray}
	where $r$ is a parameter, $r_c$ is a critical point and $C$ is a constant.
	Therefore, the both critical exponents $\delta$ at the egde of $0<\alpha<1$ 
	of the scaling behavior of Lyapunov exponents are $\displaystyle \frac{1}{2}$.
	The Floquet multiplier for the fixed point for $\alpha=0$ is -1 and for $\alpha=1$ is 1. 
	Thus, our analytic results agree with the prediction by Pomeau and Manneville \cite{Pomeau} that critical exponents of Lyapunov exponents
	of \textit{Type 1} (Floquet multiplier is unity) and \textit{Type 3} (Floquet multiplier is $-1$) intermittency are $\frac{1}{2}$ respectively.
	The critical exponent in the limit of $\alpha \to 1+0$ is different from that of $\alpha \to 1-0$. 
	This remarkable difference in critical exponents of Lyapunov exponents is caused by the difference between 
	the mixing (ergodic) phase for $0<\alpha\leq 1$ and
	the dissipative (non-ergodic) phase for $\alpha>1$.
}

\paragraph{5. Conclusion}
We obtain the analytic formula of Lyapunov exponents for the generalized Boole transformations for the full range of parameters
{and prove the mixing of the map for $0<\alpha<1$. }
We also obtain the KS-entropy by using the Pesin's formula.
As a result, we can explicitly modulate Lyapunov exponent {$\lambda$ in the range of $\lambda \geq0$ by changing parameter $\alpha>0$.}
In addition, the absolute values of the derivative $d\lambda/d\alpha$
at $\alpha=0$ to $\alpha=1-0$ are proven to be {\it infinite}. Then, it is analytically shown to be extremely difficult to distinguish values of Lyapunov exponent
near $\alpha=0$ or $\alpha=1$ because of such strong parameter dependency. 
The scaling behavior of our analytic  Lyapunov exponents of the Boole transformations is found to be consistent with  Pomeau and Manneville's numerical results about intermittency {of \textit{Type 1} and \textit{Type 3}}. 
Thus, we expect that our analytic analysis can be useful for investigating  physically chaotic systems further.


%
\section*{Acknowledgement}
One of the author, Ken-ichi Okubo would like to express his sincere gratitude to Mr. Atsushi Iwasaki at Kyoto University for his wholesome advice.

\appendix

\section{Appendix}

\textit{Proof of Theorem 2}: 
We obtain the relation of Eq. (\ref{modified Generalized Boole}) by substituting $x_n$ by $-\cot \pi\theta_n$ as, 
\begin{eqnarray}
	\cot \pi\theta_{n+1} &=& 2\alpha \cot 2\pi\theta_n, \nonumber\\
	{\rm where}~~~~~~~~~~~~\theta_{n+1} &=& \bar{T}_{\alpha}(\theta_n)=\frac{1}{\pi}{\rm arccot} \left[2\alpha \cot 2\pi\theta_n\right]. \label{cot map}
\end{eqnarray}
The transformations
\begin{eqnarray}
	\bar{T}_{\alpha} : [0, 1) \to [0, 1),
\end{eqnarray}
{have the topological conjugacy relation with Eq. (\ref{modified Generalized Boole})}. 
Consider the following cylinder sets of $\bar{T}_\alpha^{-k}[0, 1)$ on which $\bar{T}_\alpha^k$ is surjective to $[0, 1)$, satisfying
\begin{eqnarray}
	I_{j, k} &=& [\nu_{j, k}, \nu_{j+1, k}),~ \nu_{j,k}<\nu_{j+1,k}, ~\mbox{for}~ 0\leq j\leq 2^k-1,\\
	\nu_{0, k}&=& 0~ \mbox{and}~ \nu_{2^k, k}=1,\\
	\bar{T}_{\alpha}^k I_{j,k} &=& [0, 1),\\
	\mu(I_{j,k}) &=& \frac{1}{2^k}.
\end{eqnarray}
For any measurable set $A$, and $\bar{T}_\alpha$-invariant measure $\mu$, consider, for arbitrary $n(\geq k)$
\begin{eqnarray}
	\mu\left(\bar{T}_\alpha^{-n}A \cap B\right),
\end{eqnarray}
where it is enough to assume $B$ is $\left[\nu_{j, k}, \nu_{j+1,k}\right)$  without loss of generality. 
Here, $\mu(B)= \frac{1}{2^k}$ and the number of the intervals of $\bar{T}_\alpha^{-n}A$ 
which are included by $B$ is $2^{n-k}$. Thus, the measure of $\bar{T}_\alpha^{-n}A$ on each interval $I_{j,k}$
whose measure $\mu(I_{j,k})$ is \textit{equidistributed}, is $2^{-n}\mu(A)$. Then, for arbitrary $n(\geq k)$,
\begin{eqnarray}
	\mu\left(\bar{T}_\alpha^{-n}A \cap B\right) = 2^{n-k}\cdot (2^{-n}\mu(A)) = \mu(A)\cdot \mu(B)
\end{eqnarray}
Thus, the dynamical system $([0, 1), \bar{T}_\alpha, \mu)$ has mixing property according to the Ref. \cite{アーノルド}. 
Therefore, the generalized Boole transformations for $0<\alpha<1$ {have mixing property}.	\qed

\begin{figure}[!h]
	\centering
	\begin{tabular}{cc}
		\begin{minipage}{0.45\hsize}
			\vspace*{0.2cm}
			\includegraphics[width= 1.1\columnwidth]{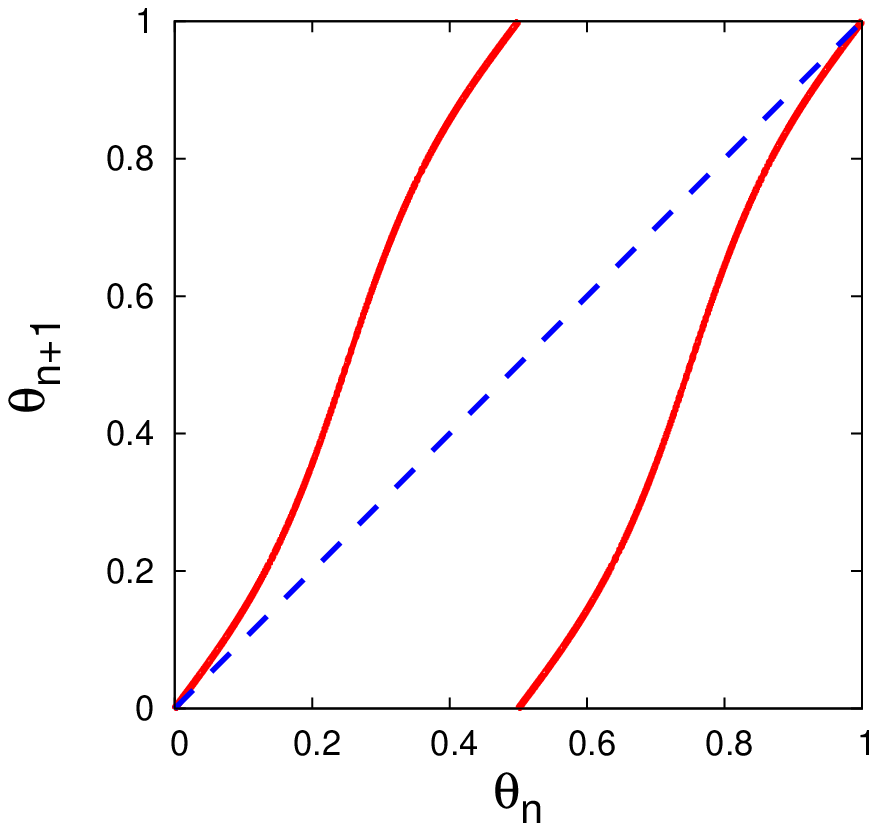}
			\vspace*{1.5cm}
			\caption{Solid lines correspond to the transformation $\bar{T}_{0.75}$ which is topologically conjugate with the generalized Boole transformation 
				$T_{\alpha=0.75, \beta=0.75}$ and a dashed line corresponds to $\theta_{n+1}= \theta_n$. 
				The transformations $\bar{T}_{\alpha}$ monotonically increase in the interval $[0, 0.5)$ or $[0.5, 1)$.}
		\end{minipage}
		\hspace*{5mm}
		
		\begin{minipage}{0.45\hsize}
			\includegraphics[width= .9\columnwidth]{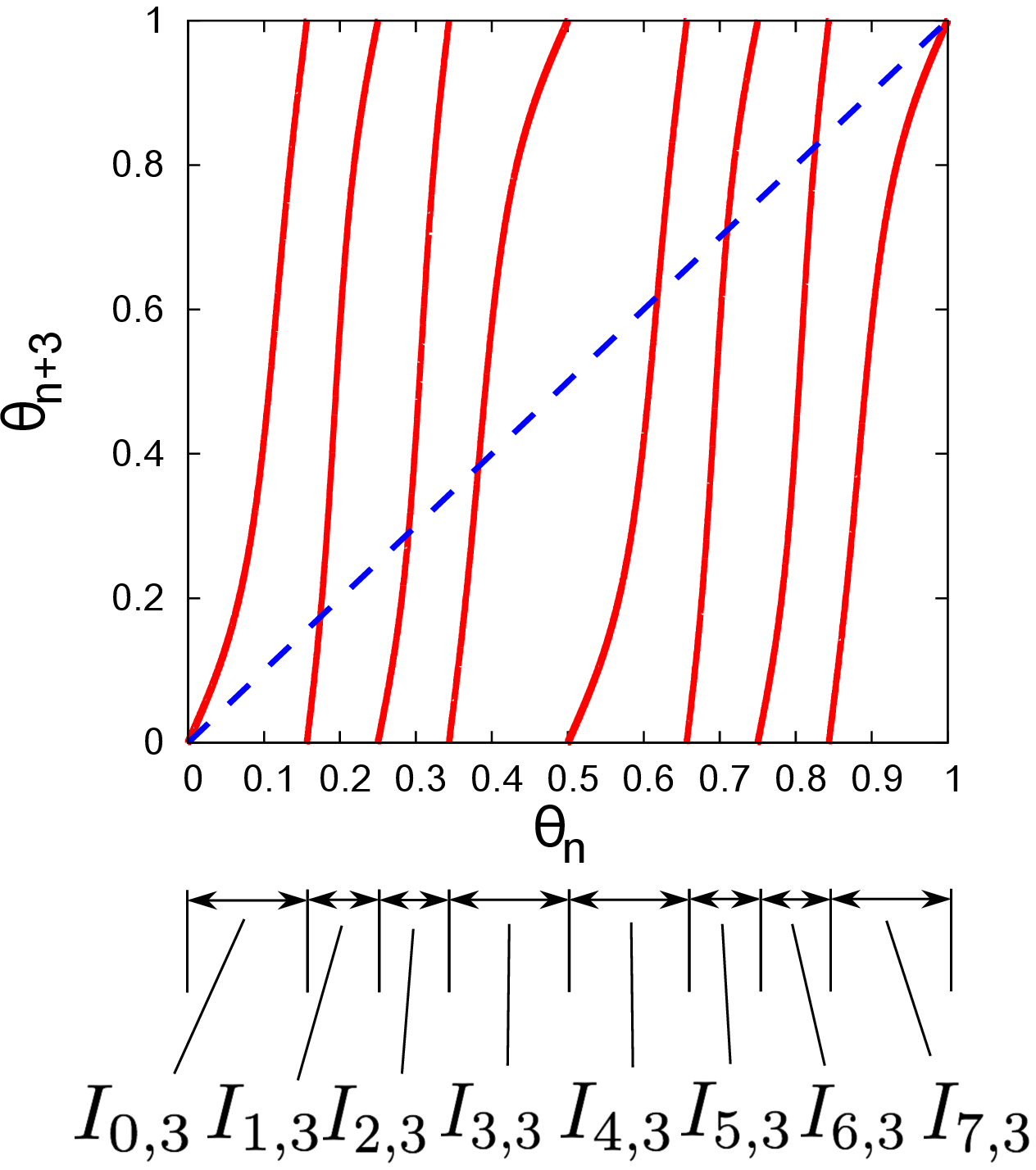}
			\caption{Solid lines correspond to the transformation $\bar{T}_{0.75}^3$ which is topologically conjugate with the generalized Boole transformation 
				$T_{\alpha=0.75, \beta=0.75}^3$ and a dashed line corresponds to $\theta_{n+3}= \theta_n$. Here, the measure of each $I_{j,3} (0\leq j \leq 7)$ is equidistributed.}
		\end{minipage}
	\end{tabular}
\end{figure}

\end{document}